\documentclass[letterpaper, 10 pt, conference]{ieeeconf}
\usepackage[utf8]{inputenc}
\usepackage[pdftex]{graphicx}
\usepackage[cmex10]{amsmath}
\usepackage{amsfonts}
\usepackage{subcaption}
\usepackage{comment}
\usepackage{float}
\usepackage[dvipsnames]{xcolor}
\usepackage[hyphens]{url}
\usepackage{hyperref}

\usepackage{comment}
\usepackage{mathtools}
\newtheorem{theorem}{Theorem}

\let\norm\undefined 
\DeclarePairedDelimiter\norm{\lVert}{\rVert}

\title{Stability and Robustness of a Hybrid Control Law \\ for the Half-bridge Inverter}
\author{Gabriel E. Colón-Reyes, Kaylene C. Stocking, Duncan S. Callaway, Claire J. Tomlin}
\date{January 2022}

\begin{document}

\maketitle

\begin{abstract}
Hybrid systems combine both discrete and continuous state dynamics. Power electronic inverters are inherently hybrid systems: they are controlled via discrete-valued switching inputs which determine the evolution of the continuous-valued current and voltage state dynamics. 

Hybrid systems analysis could prove increasingly useful as large numbers of renewable energy sources are incorporated to the grid with inverters as their interface. In this work, we explore a hybrid systems approach for the stability analysis of power and power electronic systems. We provide an analytical proof showing that the use of a hybrid model for the half-bridge inverter allows the derivation of a control law that drives the system states to desired sinusoidal voltage and current references.
We derive an analytical expression for a global Lyapunov function for the dynamical system in terms of the system parameters, which proves uniform, global, and asymptotic stability of the origin in error coordinates. Moreover, we demonstrate robustness to parameter changes through this Lyapunov function. We validate these results via simulation.

Finally, we show empirically the incorporation of droop control with this hybrid systems approach. In the low-inertia grid community, the juxtaposition of droop control with the hybrid switching control can be considered a grid-forming control strategy using a switched inverter model. 

\end{abstract}

\section{Introduction}
The power grid is seeing a large portion of its generation replaced by renewable energy. Therefore, the energy conversion process is changing and converter-interfaced generation (CIG) systems, primarily in the form of power electronic inverters, are becoming the primary interface between energy sources and loads.

\subsection{The Need for New Inverter Controls}

Inverters are highly controllable devices that convert a DC energy source, like a solar photovoltaic panel, to a grid-compatible AC energy form. Fig. \ref{fig:header_} presents the standard interconnection of a solar photovoltaic panel with the grid through an inverter, where we show the hybrid control strategy interfacing a grid-forming control strategy with the inverter. Our contributions focus on this hybrid control strategy. 

Since the number of inverters in the power grid is increasing, and because they react to programmable instructions, the design of control strategies
for power grids with large numbers of CIGs is considered a problem of paramount importance in the low-inertia grid research community \cite{lin2020research}. Emphasizing this point, the European Union funded in 2016 the MIGRATE (\textbf{M}assive \textbf{I}nte\textbf{grat}ion of Power \textbf{E}lectronic Devices) project \cite{MIGRATE} to study fundamental challenges associated with incorporating renewables in the grid. The team was composed of experts from over 20 participating institutions from industry and academia. Furthermore, in 2021, the United States of America's Department of Energy funded the UNIFI (\textbf{Uni}versal Interoperability for Grid‐\textbf{F}orming \textbf{I}nverters) Consortium \cite{UNIFI}, consisting of experts from over 40 participating institutions, to answer questions of a similar nature.

\begin{figure}[t]
\begin{centering}
    \includegraphics[scale=1.25]{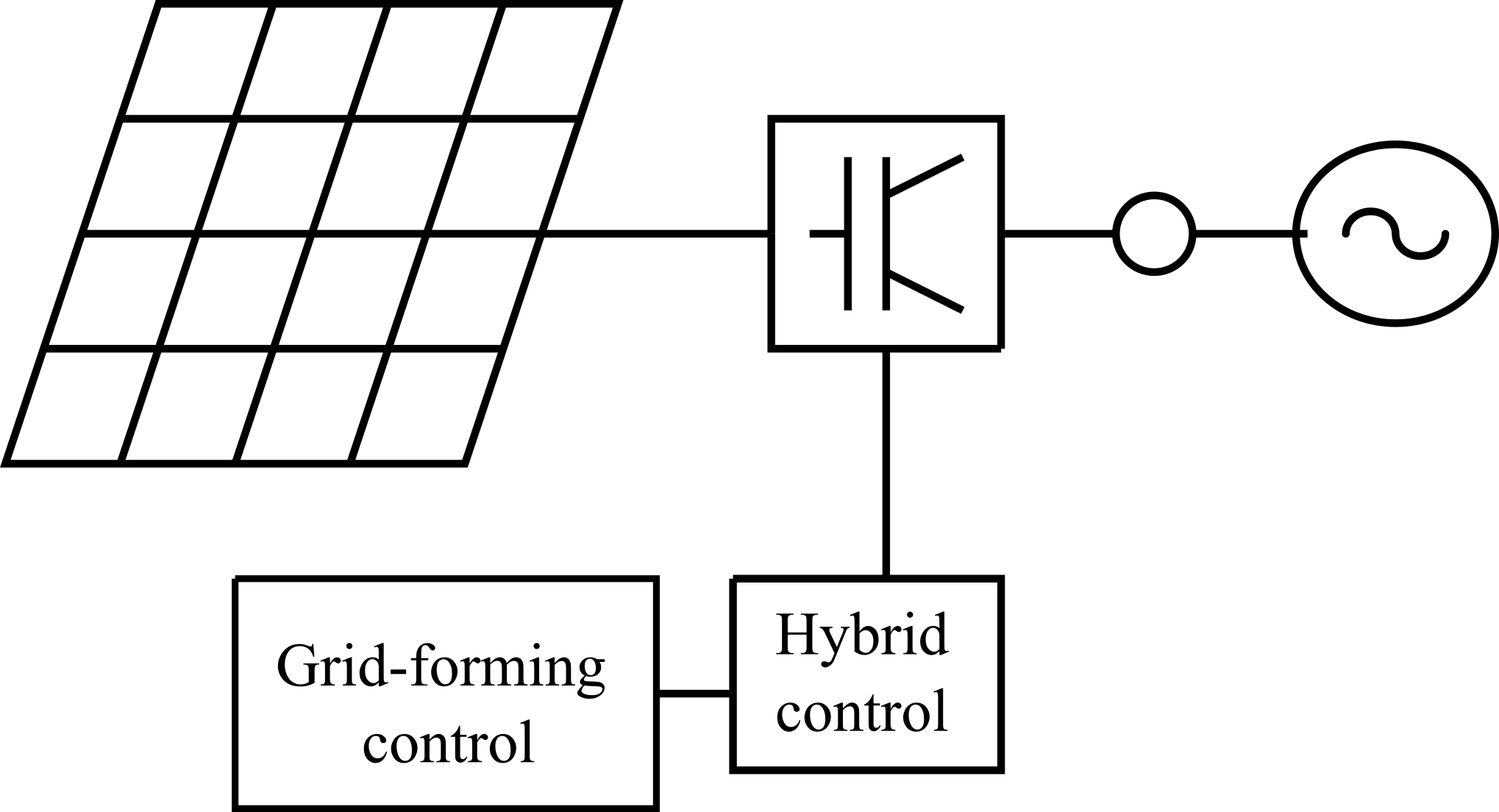}
    \caption{Typically, a DC energy source, such as a solar photovoltaic panel (shown on the far left), powers an inverter, which is controlled via a grid-forming control strategy. We present where our hybrid control strategy would find itself in the context of an interconnection with the grid.} 
    \label{fig:header_}
\end{centering}
\end{figure}

Recent work further recognizes the need to consider new control methodologies for inverters \cite{milano2018foundations}. Accentuating this need, the IEEE task force for stability definitions and classification for low-inertia grids has specifically recognized the importance of hybrid systems concepts for the study of power and power electronic systems \cite{hatziargyriou2020stability}.
Hence, due to the importance that inverter control strategies will play in the future grid, we believe it is imperative to explore the whole landscape of possible controllers, including those that, as in this work, rely on the more physically realistic hybrid systems model of the inverter.

One approach to inverter control is the grid-forming paradigm, in which the inverter is configured as a controllable voltage source with provisions that enable power sharing in parallel operation. A future grid will likely need a large portion of its generation controlled via grid-forming controls \cite{lin2020research}.

Several grid-forming control strategies have been proposed with strong theoretical guarantees and validation through simulation and experimentation \cite{johnson2013synchronization, beck2007virtual, arghir2018grid, markovic2021understanding}. However, their implementations thus far use pulsed-width modulation (PWM) strategies whose analysis rely on the inverter averaged model and lead to the well-studied nested voltage and current feedback loop structure. The averaged model constrains the kinds of control methods that can be implemented by providing a template for which grid-forming control strategies need to fit within, and prevents analysis of very fast control dynamics at the system level.

\subsection{Hybrid Systems for Power and Power Electronic Systems}

The field of hybrid systems has only seen limited use in the context of the stability analysis and control design for power and power electronic systems. 

In the power systems literature, most of the known work has attempted to model general power system problems within a hybrid dynamical systems framework. For example, some applications include the modeling of continuous-valued state variables (e.g., currents and voltages) in systems with discrete-valued control or protection actions (e.g., circuit breaker operation) \cite{hiskens2004power}. Other work has considered the variation of system-level inertia in the context of the swing equations and modeled those dynamics in a hybrid dynamical systems framework \cite{hidalgo2018frequency}.

In the power electronics literature, hybrid systems theory has been primarily used in the context of DC-DC converters. One of the first reported uses of a hybrid systems framework for power electronic converters was in \cite{senesky2003hybrid}, where the authors derive conditions for a safe set within which a designed hybrid control law is able to keep the system states. Further work derives a set of control laws, also for DC-DC converters, that achieve global, asymptotic stability to the desired operating point \cite{buisson2005stabilisation}.

More recently, some of these ideas have started to appear in the control of power electronic inverters, for which the desired reference signal is a time-varying sinusoid instead of a constant setpoint. To the best of our knowledge,  \cite{albea2017hybrid} is the first attempt to have an inverter track a sinusoidal reference with a hybrid control law. The authors of \cite{albea2017hybrid} derive a model for the half-bridge inverter in error coordinates and show that the solution of an optimization problem leads to the desired behavior.

\subsection{Summary of Contributions}



Specifically, beyond prior results, this work's contributions include i.) an alternate  exposition and proof of the hybrid control method in \cite{albea2017hybrid}, including an explicit, closed-form expression for the control law not previously presented, ii.) an analytical derivation for a global Lyapunov function that proves uniform, global, asymptotic stability of the origin in error coordinates thus achieving the tracking of the reference signal, iii.) a method to update the controller to be robust against resistive load changes, iv.) analysis and validation of the controller in simulation, and v.) a demonstration in simulation of the controller's operation in conjunction with a grid-forming control strategy, suggesting the potential of hybrid systems-based control in this practical context.




\subsection{Paper Organization}

The rest of the paper is organized as follows. Section \ref{sec: prob form} presents the notation and problem formulation, Section \ref{sec: theory} presents our theoretical results, and Section \ref{sec: experiments} our experimental results. Finally, Section \ref{sec: conclusion} concludes the paper and discusses limitations and future work.

\section{Problem Formulation}
\label{sec: prob form}
In this section, we present our modeling approach for the half-bridge inverter and the control problem we aim to solve. The model we use was first proposed in \cite{albea2017hybrid}.

\subsection{Notation}

Dot notation indicates the time derivative of a variable, i.e., $\frac{dx}{dt} = \dot{x}$. Bold-faced capital letters will indicate matrices. $\sigma(\bold{A})$ is the spectrum of $\bold{A}$, and $\lambda_\bold{A}$ an eigenvalue of $\bold{A}$. $\norm{\cdot}_2$ is the Euclidean vector norm. The operator $sign(\cdot)$ returns $+1$ if its scalar argument is greater than or equal to zero, and $-1$ otherwise. The operator $Re\{ \cdot \}$ takes the real part of its argument.

\subsection{Switched Model}

\begin{figure}[t]
\begin{centering}
    \includegraphics[scale=0.9]{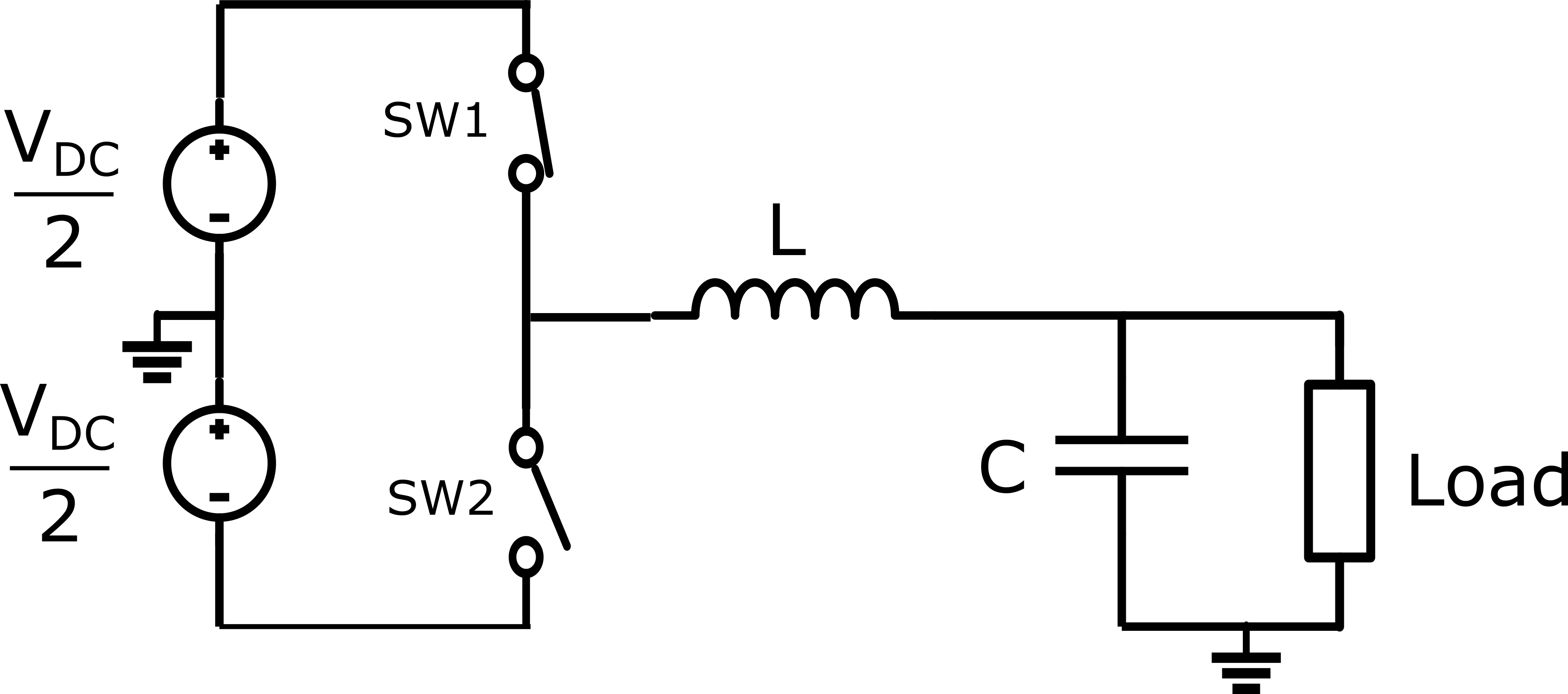}
    \caption{The half-bridge inverter is composed of a midpoint-grounded voltage source $V_{DC}$, a pair of switches, and an LC filter. It delivers power to an arbitrary load.}
    \label{fig:half_bridge}
\end{centering}
\end{figure}

We study the half-bridge inverter shown in Fig. \ref{fig:half_bridge}. We consider a resistive load, $R$, as the natural first case study. We define two discrete states: i.) SW1 on and SW2 off, and ii.) SW1 off and SW2 on. We also define the control input to be $u \in \{+1, -1\}$, where $u = +1$ denotes the first discrete operating state, and $u = -1$ denotes the second. We define the state vector as $x = \begin{bmatrix} v_C & i_L \end{bmatrix}^\top \in \mathbb{R}^2$, where $v_C$ and $i_L$ are the capacitor voltage and inductor current, respectively.
Applying Kirchhoff's current and voltage laws to model the voltage and current dynamics for each discrete state, we obtain equation \eqref{eq: switched_dynamics}.  Conveniently, the choice of control input, $u \in \{ +1,-1\}$, characterizes the inverter's terminal voltage polarity in Kirchhoff's voltage law equations such that the dynamic evolution of both discrete states can be compactly characterized by a single dynamical system. The model describes a linear, time-invariant (LTI) dynamical system with a binary-valued, discrete control input.
\begin{subequations}
\begin{align}
\label{eq: switched_dynamics}
\dot{x} &= \bold{A}x + \bold{B}u \\
\bold{A} &= \begin{bmatrix} -\frac{1}{RC} & \frac{1}{C} \\ -\frac{1}{L} & 0 \end{bmatrix}, \; \bold{B} = \begin{bmatrix} 0 \\ \frac{V_{DC}}{2L} \end{bmatrix}
\end{align}
\end{subequations}

\subsection{Designing the Reference}
The control objective is for $v_C$ to track a sinusoidal reference signal. Following \cite{albea2017hybrid}, we first define an oscillator system that produces the appropriate reference signal. Then, we reformulate the dynamical system in terms of error from the desired reference state.

We denote the reference vector $x_{ref}=\begin{bmatrix} v_{C,ref} & i_{L,ref} \end{bmatrix}^{\top}$. We choose a sinusoidal voltage reference $v_{C,ref}(t) = V_m \text{sin}(\omega t)$, with amplitude $V_m$ and frequency $\omega$, both being design parameters. Kirchhoff's laws allow us to express the appropriate inductor current reference $i_{L,ref}$ in terms of these parameters as $i_{L,ref}(t) = \omega C V_m \text{cos}(\omega t) + \frac{1}{R} V_m \text{sin}(\omega t)$.

To design the reference, we define an oscillator system of the form of \eqref{eq: osc_dynamics}.
\begin{align}
\label{eq: osc_dynamics}
\dot{z} = \bold{\Theta} z = \begin{bmatrix} 0 & \omega \\ -\omega & 0 \end{bmatrix}z
\end{align}
Note that $\bold{\Theta}$ is a skew-symmetric matrix with eigenvalues $\lambda_{\bold{\Theta}} \in \{+j\omega, -j\omega\}$ defining an oscillator. Therefore, we know that, for a given initial condition, $z(t_0)$, at time $t_0$ satisfying $$ z(t_0) \in \{ z \in \mathbb{R}^2 \; |  \; \norm{z}_2^2 = V_m^2 \},$$ $z(t) = \begin{bmatrix} V_m \text{sin}(\omega t) & V_m \text{cos}(\omega t) \end{bmatrix}^{\top}$. 

We define our state reference to be a linear transformation of the oscillator state, $z$, to achieve $x_{ref}$ by
\begin{align}
\label{eq: reference}
x_{ref} = \bold{\Pi} z = \begin{bmatrix} 1 & 0 \\ \frac{1}{R} & \omega C \end{bmatrix}z.
\end{align}
We further compute the reference dynamics to arrive at \eqref{eq: reference_dynamics_c}.
\begin{subequations}
\begin{align}
\label{eq: reference_dynamics_a}
\dot{x}_{ref} &= \bold{\Pi} \dot{z} \\
\label{eq: reference_dynamics_b} &= \bold{\Pi} \bold{\Theta} z \\
\label{eq: reference_dynamics_c} &= (\bold{A} \bold{\Pi} + \bold{B} \bold{\Gamma}) z,
\end{align}
\end{subequations}
where
\begin{align}
\label{eq: Gamma}
\bold{\Gamma} = \frac{2}{V_{DC}} \begin{bmatrix} \omega \frac{L}{R}  & (1 - \omega ^2 LC) \end{bmatrix} = \begin{bmatrix} \Gamma_1 & \Gamma_2 \end{bmatrix}.
\end{align}

Our subsequent analysis will be simplified by defining the state error from the reference signal as $e=x-x_{ref}$. Combining equations \eqref{eq: switched_dynamics} and \eqref{eq: reference_dynamics_c}, we have that 
\begin{subequations}
\begin{align}
\label{eq: error_dynamics_a}
\dot{e} &= \bold{A}x + \bold{B}u - (\bold{A} \bold{\Pi} z + \bold{B} \Gamma z) \\
\label{eq: error_dynamics_b}
&= \bold{A}e + \bold{B}(u - \bold{\Gamma} z).
\end{align}
\end{subequations}
Together, equations \eqref{eq: osc_dynamics} and \eqref{eq: error_dynamics_b} define the system dynamics.

Now that we have fully characterized the system dynamics in error coordinates, the problem we are trying to solve is to find a control law $u(t)$ that drives the error coordinates to the origin, implying that the system states will track the desired reference signals.

\section{Theoretical Results for Global Asymptotically-Stable Reference Tracking}
\label{sec: theory}
In this section, we present our theoretical stability results including i.) a proof for deriving an explicit control law that drives the error dynamics to the origin, ii.) an analytical solution to the Lyapunov equation as a function of the system's parameters, and iii.) a proof showing that we can appropriately adjust our control law to be robust against known changes in load, and still track the desired trajectory.

\subsection{Global Asymptotic Stability Result}
\label{sec: switching result}


\begin{theorem}
\label{thm:stability}
Let $\bold{A}$ be Hurwitz, i.e. $Re\{\lambda_\bold{A}\} < 0 \; \; \forall \lambda_\bold{A} \in \sigma(\bold{A})$, and let $\norm{\bold{\Gamma}}_2 < \frac{1}{V_m}$. Then, the switching policy
\begin{align}
\label{eq: policy}
   u = -sign(\bold{B}^{\top}\bold{P}e)
\end{align} 
results in the uniform, global, asymptotic stability of the origin for the error dynamics \eqref{eq: error_dynamics_b}.
\end{theorem}

\begin{proof}
Let $\bold{P} \in \mathbb{R}^{2 \times 2}$ be the symmetric, positive-definite matrix that satisfies the Lyapunov equation, $\bold{A}^{\top} \bold{P} + \bold{P}\bold{A} = \bold{Q}$. We choose $\bold{Q} = -\bold{I}$, where $\bold{I}$ is the $2 \times 2$ identity matrix. Note that $\sigma (-\bold{I}) = \{-1, -1\}$. The existence of such a $\bold{P}$ is guaranteed by the assumption that $\bold{A}$ is Hurwitz \cite{callier2012linear}. We propose the candidate Lyapunov function $V(e) = e^{\top}\bold{P}e$, which is globally positive definite \cite{sastry2013nonlinear}. Taking its time derivative, we have that
\begin{align}
\dot{V}(e) &= e^{\top}(\bold{A}^{\top}\bold{P} + \bold{P}\bold{A})e + 2(u - \bold{\Gamma} z)\bold{B}^{\top}\bold{P}e.
\end{align}

Furthermore, using our knowledge of the oscillator state, $z(t)$, we have that
\begin{subequations}
\begin{align}
\label{eq: Gamma z bound_a}
\bold{\Gamma} z &= V_m (\Gamma_1 \text{sin}(\omega t) + \Gamma_2 \text{cos}(\omega t) ), \\ 
\label{eq: Gamma z bound_b}
&= V_m\norm{\bold{\Gamma}}_2\text{sin}(\omega t + \psi), \; \psi = \text{arctan}(\frac{\Gamma_2}{\Gamma_1}), \\
\label{eq: Gamma z bound_c}
\Rightarrow |\bold{\Gamma} z | &\leq V_m\norm{\bold{\Gamma}}_2, \;\; \forall t.
\end{align}
\end{subequations}
Therefore, the worst case magnitude for $u - \Gamma z$, i.e., $\text{max.}_z \, \{u -\bold{\Gamma} z \}$, is obtained when $\text{sin}(\omega t + \psi) = -1$. So, we can upper bound $\dot{V}(e)$ using the Lyapunov equation and \eqref{eq: Gamma z bound_c} to yield the following. 
\begin{align}
\label{eq: lyapunov_inequality}
\dot{V}(e) \leq e^{\top}(-\bold{I})e + 2(u +V_m\norm{\bold{\Gamma}}_2)\bold{B}^{\top}\bold{P}e
\end{align}

We know the term $-e^{\top}e \leq 0, \; \forall e$. Therefore, to ensure stability, we force the second term, $2(u + V_m\norm{\bold{\Gamma}}_2)\bold{B}^{\top}\bold{P}e$, to also be less than or equal to zero using the following. Since by assumption $\norm{\bold{\Gamma}}_2 < \frac{1}{V_m}$, and $u \in \{+1, -1\}$ by design, then the term $u + V_m\norm{\bold{\Gamma}}_2$ will take the sign of $u$. By choosing a switching policy defined as $u = -sign(\bold{B}^{\top}\bold{P}e)$ under the stated assumptions, the second term will always be less than or equal to zero, and $\dot{V}(e) \leq 0$ always holds, implying that $V(e)$ is a valid, global Lyapunov function for the given dynamics. We can thus conclude that the origin in error coordinates is a uniformly-, globally-, asymptotically-stable equilibrium point. The states will track the desired trajectories in the original coordinates. \end{proof}

It is worth noting that the assumptions required for this theorem are not very restrictive. For a passive RLC circuit as we have here, it can be proven that the eigenvalues of $\bold{A}$ will always have strictly negative real parts. Moreover, $\norm{\bold{\Gamma}}_2 < \frac{1}{V_m}$ is easily satisfied for a range of realistic parameter values (the magnitude and frequency of the reference, the filter parameters, the load resistance, and the input voltage, as given by \eqref{eq: Gamma}).

It is also worth noting that the control policy \eqref{eq: policy} is conservative: \eqref{eq: lyapunov_inequality} is a worst-case upper bound for the Lyapunov function. This implies that it is, potentially, inherently robust to some disturbances.

\subsection{Explicit Lyapunov Function in Terms of R, L, and C}
\label{sec:P_matrix}


In general, finding an explicit Lyapunov function that guarantees stability of a hybrid dynamical system can be difficult, even if each discrete dynamic state is LTI. To do this, researchers often resort to semidefinite optimization programs \cite{borrelli2017predictive}. However, since the error coordinate dynamics (\ref{eq: error_dynamics_b}) reduce our system model to a single set of LTI dynamics, we can calculate the appropriate $\bold{P}$ matrix that satisfies the Lyapunov equation, $\bold{A}^{\top} \bold{P} + \bold{P}\bold{A} = \bold{Q}$, by solving a system of linear equations. 

In this section, we show that this procedure allows us to express the $\bold{P}$ matrix explicitly in terms of arbitrary load and inverter parameters, $R$, $L$, and $C$.

Noting that $\bold{P}$ is symmetric, we have that
\begin{equation}
    \bold{P} = \begin{bmatrix} p_{11} & p_{12} \\ p_{21} & p_{22} \end{bmatrix},
\end{equation}
where $p_{12} = p_{21}$.

We parameterize $\bold{Q} = -\alpha \bold{I}$ with a coefficient $\alpha$ to allow for performance tradeoffs. Therefore, the matrix equation $\bold{A}^{\top} \bold{P} + \bold{P}\bold{A} = -\alpha \bold{I}$ evaluates to
\begin{equation}
    \begin{bmatrix} \frac{-2p_{11}}{RC} - \frac{2p_{12}}{L} & \frac{p_{11}}{C} - \frac{p_{12}}{RC} - \frac{p_{22}}{L} \\ \frac{p_{11}}{C} - \frac{p_{12}}{RC} - \frac{p_{22}}{L} & \frac{2p_{12}}{L} \end{bmatrix} = -\alpha \begin{bmatrix} 1 & 0 \\ 0 & 1 \end{bmatrix}.
\end{equation}

We can use this equality to solve for $p_{11}$, $p_{12}$, and $p_{22}$. Doing so gives us an explicit solution for $\bold{P}$ in \eqref{eq: P}, as desired.
\begin{equation}
\label{eq: P}
    \bold{P} = \frac{\alpha}{2}\begin{bmatrix} RC + \frac{RC^2}{L} & -C \\ -C & RL + \frac{L}{R} + RC \end{bmatrix}
\end{equation}

Therefore, a function $V(e) = e^\top \bold{P}e$, with $\bold{P}$ given by \eqref{eq: P}, is guaranteed to be a global Lyapunov function for the dynamical system \eqref{eq: error_dynamics_b}.

\subsection{Controlling for Known Changes in Load}

One motivation for finding an analytical expression for $\bold{P}$ is to handle changes in the system parameters. Note that our control policy, \eqref{eq: policy}, relies on knowing $\bold{P}$. For example, if the load resistance, $R$, or filter parameters, $L$ or $C$, change, the $\bold{A}$ matrix that defines the inverter dynamics in \eqref{eq: switched_dynamics} will change accordingly, and the switching strategy derived in section \ref{sec: switching result} may no longer hold if we do not update $\bold{P}$ to reflect that change.

We have shown that we can use \eqref{eq: policy} to asymptotically drive the states to our desired sinusoidal reference signals. Moreover, we have shown that for a specific choice of $R$, $L$, $C$ parameters, we can analytically compute a $\bold{P}$ matrix which makes $V(e) = e^\top \bold{P}e$ a global Lyapunov function for the dynamical system, and thus we can implement our control law with such a $\bold{P}$.

Based on these two arguments, we now make the claim that if the system load changes, that is, $R$ changes from, for example, $R_1$ to $R_2$, and we know both loading scenarios, we can update our control law, \eqref{eq: policy}, through \eqref{eq: P} at the time the load changes. This, however, assumes that the load does not switch between different parameter values too quickly to allow the control to still asymptotically drive the error state to the origin \cite{liberzon2003switching}). This is a reasonable assumption from the power and power electronics systems perspectives.

We provide validation through simulation for these three theoretical results in the following section.



\section{Experimental Validation}
\label{sec: experiments}

In this section, we test in simulation the performance of our switching controller. 
The simulation parameters are shown in Table \ref{tab:parameters}. The inverter operates at a 1 MHz switching frequency. We assume the switches are ideal: operating instantaneously and without losses.

\begin{table}[t]
    \centering
    \begin{tabular}{c|c|c}
       \textbf{Parameter} & \textbf{Symbol} & \textbf{Value} \\ \hline
        Load resistance & $R$ & 50 $\Omega$ \\ \hline
        Inverter inductance & $L$ & 450 $\mu$H \\ \hline
        Inverter capacitance & $C$ & 2.5 mC \\ \hline
        DC supply voltage & $V_{DC}$ & 1,200 V \\
        \hline
        Target (reference) frequency & $f$ & 60 Hz \\
        \hline
        Target (reference) angular frequency & $\omega$ & $2 \pi f \frac{\text{rad}}{s}$ \\ 
        \hline
        Target (reference) magnitude & $V_m$ & 177 V \\
    \end{tabular}
    \caption{Parameter values used for simulation experiments.}
    \label{tab:parameters}
\end{table}

\begin{figure}
    \includegraphics[scale=0.7]{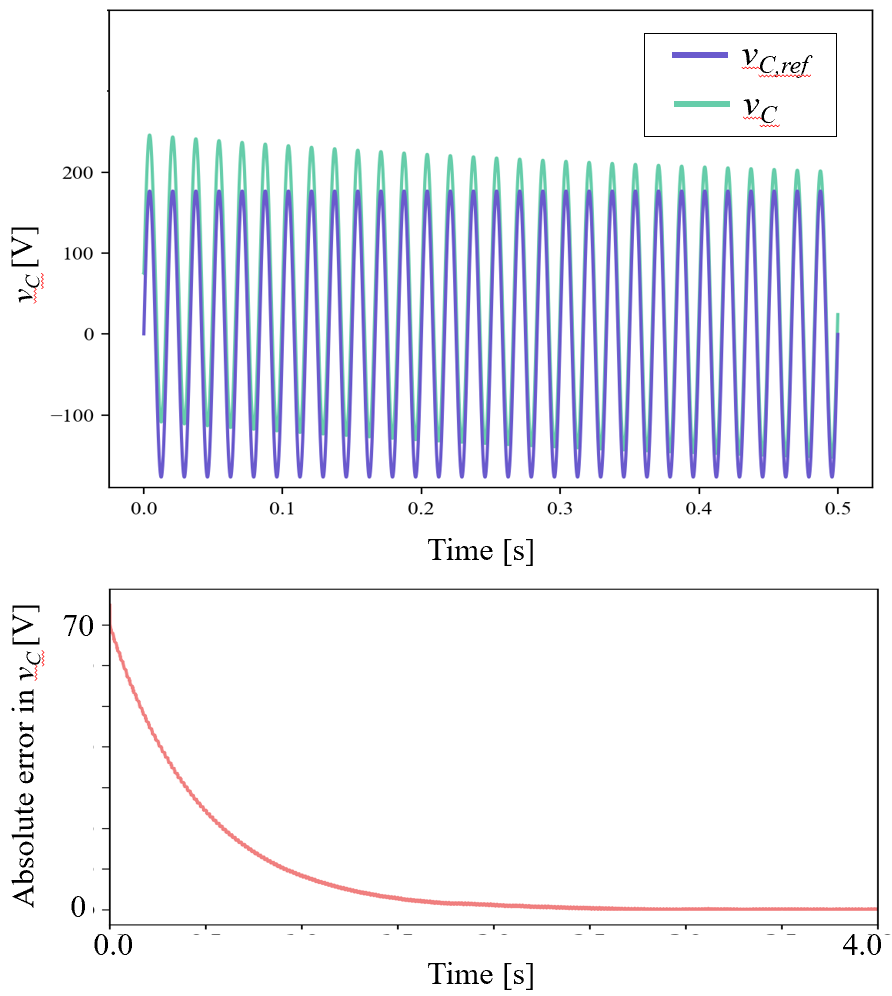}
     \caption{We show accurate tracking of the desired reference signal. In the top panel, our controller drives the inverter voltage, $v_C$, to the reference signal voltage, $v_{C,ref}$, when there is a 70 V difference in the initial condition at $t= 0$ s, over a half-second time window. The bottom panel shows the absolute error in the voltage tracking performance over a four-second time window showing convergence to the reference in steady state.}
    \label{fig:asymptotic_stability}
\end{figure}

\begin{figure*}
    \includegraphics[scale=0.54]{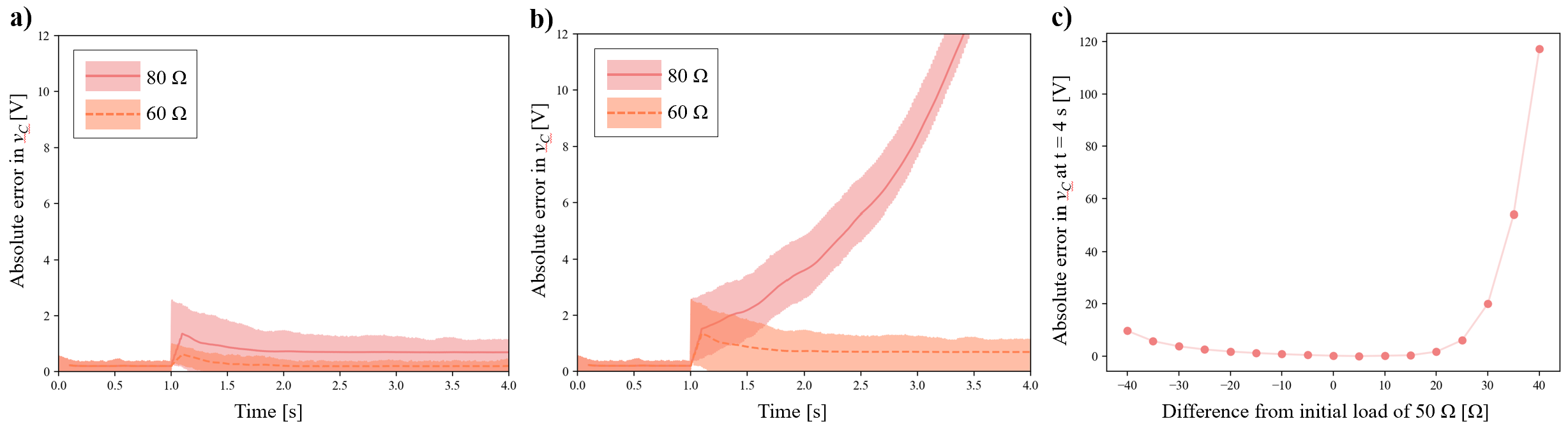}
    \caption{We show our controller's performance under known and unknown changing load conditions. a) We change our loading condition to $R = 60 \; \Omega$ and $R = 80 \; \Omega$ at $t=1.0$ s. Updating our controller in real time maintains accurate tracking of the reference signal, as evidenced by small error. b) For a change in load from 50 $\Omega$ to 60 $\Omega$ at $t = 1.0$ s, a non-updated controller still tracks the reference, although the tracking error increases slightly. A larger change to 80 $\Omega$ causes the error to diverge. c) We plot the tracking error at $ t = 4$ s for a variety of load conditions, specifically, resistance values starting at 10 $\Omega$ and ending at 90 $\Omega$, while not updating the controller to get a glimpse of its robustness. These figures suggest that the updated controller is robust to parameter changes, and the non-updated controller is also robust, but only to limited disturbances in the load conditions.}
    \label{fig:disturbed_load}
\end{figure*}

\subsection{Asymptotic Stability Under Constant Load}

We first test the performance of our controller at recovering the target reference signal when starting from an off-reference initial condition. As shown in Fig. \ref{fig:asymptotic_stability}, the capacitor voltage, $v_C$, starts with a large deviation in the initial condition of $v_C(t = 0 \; \text{s}) = 70$ V, but the controller is able to drive it to the reference over time. The controller also allow for step changes in voltage amplitude and frequency.

\subsection{Stability Under Changes in Load}

If the loading conditions change, we can ensure continued tracking of the reference by updating the control policy \eqref{eq: policy} through the $\bold{P}$ matrix used in the controller, as described in section \ref{sec:P_matrix}. We validate this result in simulation by changing our initial load of $R$ = 50 $\Omega$ to different values during the simulation. As shown in panel a) of Fig. \ref{fig:disturbed_load}, an appropriately updated controller continues tracking the reference after $R$ is increased to 60 $\Omega$ and 80 $\Omega$.

Updating the controller in response to load changes allows us to guarantee continued global, asymptotic stability. Moreover, empirically we find that for small disturbances an unmodified controller also performs well. However, stability is lost for large disturbances such as the shift to 80 $\Omega$, as shown in panels b) and c) of Fig. \ref{fig:disturbed_load}.


\subsection{Stability Under Extreme Reference Values}

Since the result of Theorem \ref{thm:stability} assumes that $\norm{\bold{\Gamma}}_2 < \frac{1}{V_m}$, and $\norm{\bold{\Gamma}}_2$ depends on $\omega$, our controller is not guaranteed to track a reference signal for all possible values of $V_m$ and $\omega$. As shown in Fig. \ref{fig:large_reference}, while good tracking performance is possible for a wide range of practical values of $V_m$ and $\omega$, increasing tracking error occurs when either value is too large. This happens because the assumptions of Theorem \ref{thm:stability} no longer hold. As can be seen in the figure, there is very good agreement between the predicted range of values where our controller should be stable and the experimentally determined range of values where we achieve stable reference tracking in practice.


\begin{figure}[!h]
    \includegraphics[scale=0.7]{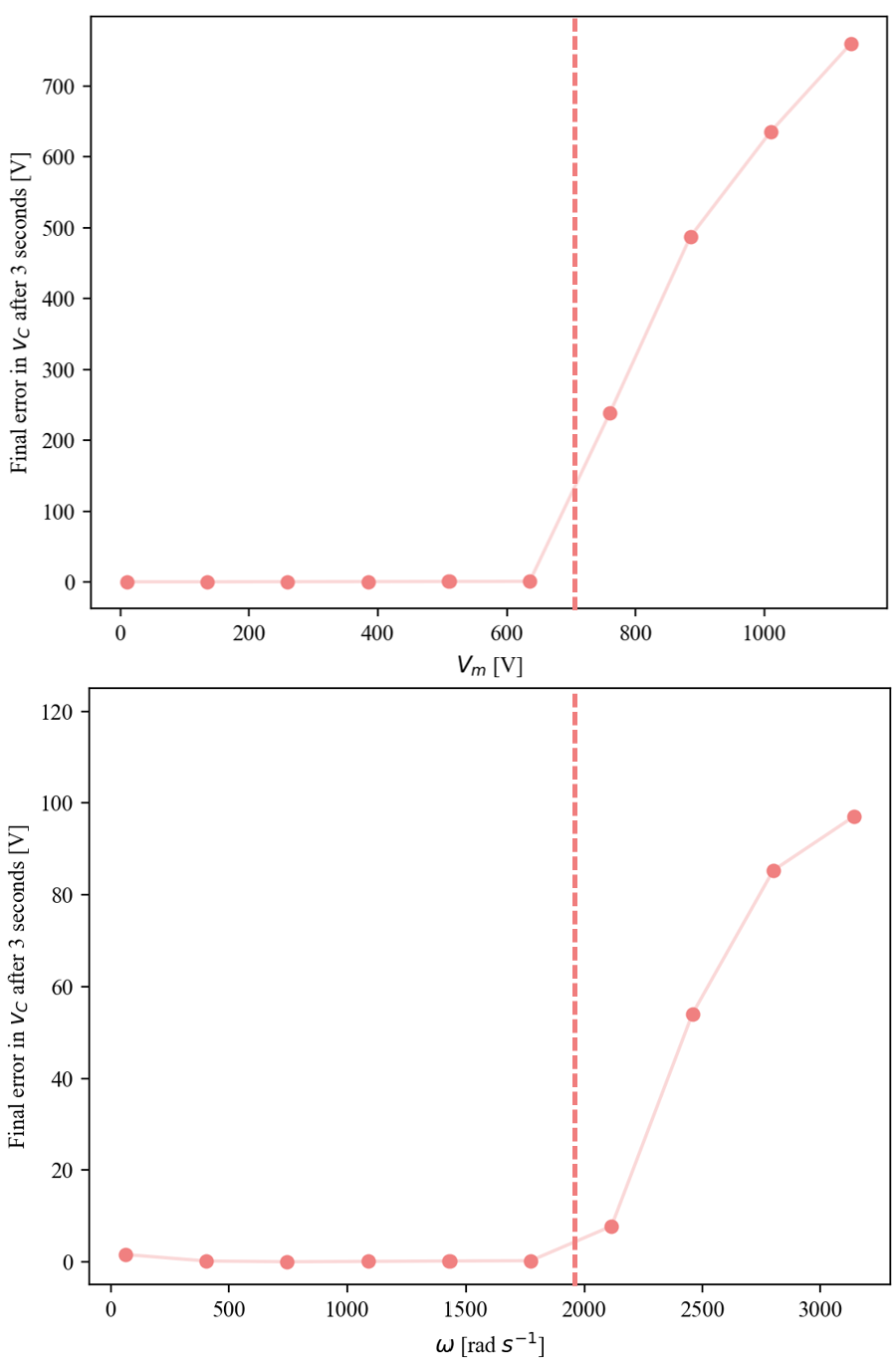}
    \caption{Increasing the magnitude, $V_m$, or the angular frequency, $\omega$, of the reference signal can prevent our controller from tracking it accurately. The top panel shows the result of increasing $V_m$ while holding $\omega$ constant at $120\pi$ rad $s^{-1}$, and the bottom panel shows increasing $\omega$ while holding $V_m$ constant at 177 V. The dotted line shows the cutoff for stability predicted by the conditions on Theorem \ref{thm:stability}; our results are consistent with the theoretical prediction.}
    \label{fig:large_reference}
\end{figure}

\subsection{Layering a Grid-Forming Control Strategy}

Until now, we have only focused on having the inverter track a sinusoid with a given amplitude and frequency. In practice, inverters will regulate their sinusoidal voltage amplitude and frequency based on a predetermined grid-forming control strategy. In general, these try to optimize for some system-level operation criteria, for example, power sharing, frequency regulation, or a stability metric. Examples in the literature include droop control, virtual oscillator control, synchronous machine emulation, and matching control, among others \cite{johnson2013synchronization, beck2007virtual, arghir2018grid, markovic2021understanding}.

These control strategies have been implemented in the past using the inverter averaged model \cite{yazdani2010voltage}. However, in this work, we use simulations to explore the performance of our hybrid control strategy while layering on top a grid-forming control strategy. We choose droop control as an example.


Droop control regulates the amplitude, $V_m$, and frequency, $\omega$, of the reference signal according to the following equations.
\begin{subequations}
\begin{align}
    \omega &= \omega^* + k_p(P^* - P) \\
    V_m &= V_m^* + k_q(Q^* - Q)
\end{align}
\end{subequations}
Here $P^*$ and $P$ represent, respectively, the setpoint and measured values for real power supplied by the inverter while $Q^*$ and $Q$ represent, respectively, setpoint and measured values for reactive power. $V_m^*$ and $\omega^*$ are the setpoint values for the reference signal, and $k_p$ and $k_q$ are the droop coefficients \cite{markovic2021understanding}. 

For our experiment, we begin with $V_m^*$ = 177 V and $\omega^* = 120 \pi \; \frac{\text{rad}}{\text{s}}$, as before. We change the steady-state real power being drawn by the resistive load as given by $$ P^* = \frac{V_{m}^{*2}}{R}$$ by changing the amplitude setpoint to $V_m^*$ = 185 V while leaving $\omega^*$ unchanged. Although Theorem \ref{thm:stability} does not guarantee stability under \textit{continuously} changing setpoints provided by droop control, we find that, for values of $k_p$ = 0.01 $\frac{rad/s}{W}$ and $k_q$ = 0.0025 $\frac{V}{VAR}$, our controller is able to track the changing reference very well, as shown in Fig. \ref{fig:droop_vc}.

\begin{figure}[!h]
    \includegraphics[scale=0.6]{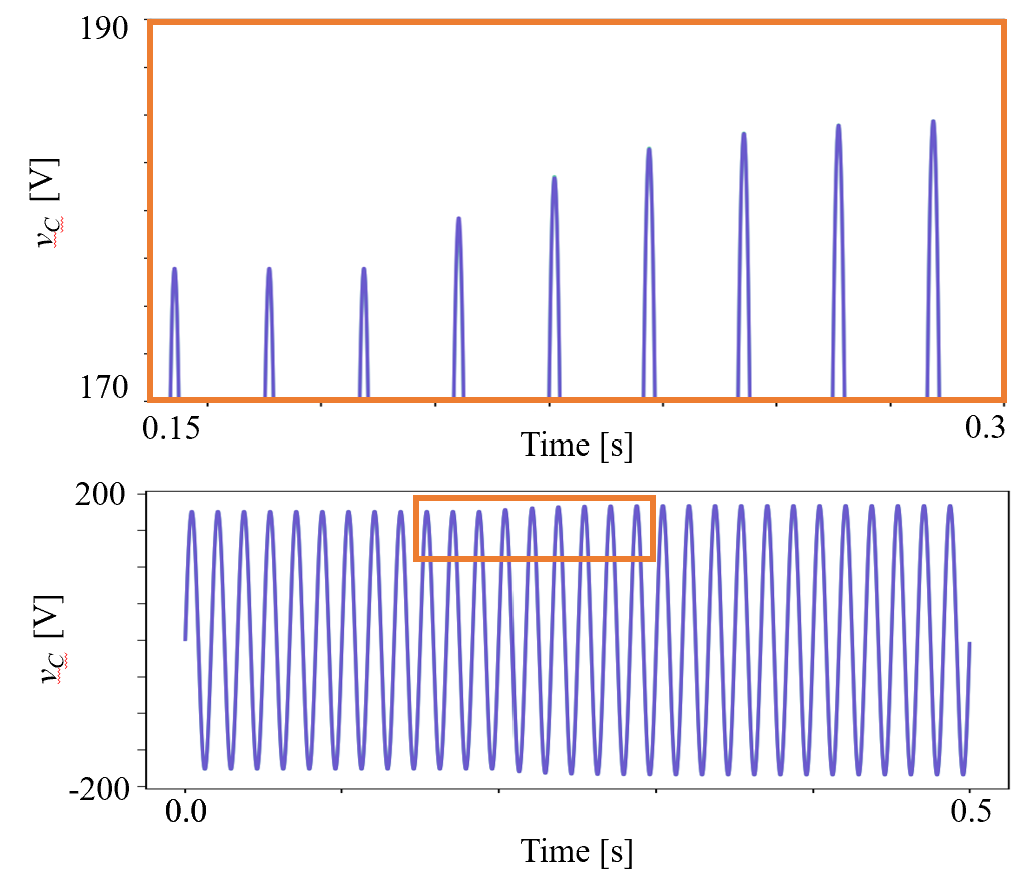}
    \caption{We show the behavior of our hybrid controller when the voltage amplitude and frequency setpoints are provided by grid-forming droop control. Our controller performs well in conjunction with droop control. In this example, the real power being drawn by the load is changed by changing the voltage magnitude setpoint, $V_m$, from 177 V to 185 V at $t=0.2$ s. The droop controller commands new setpoints for $V_m$ and $\omega$ based on the measured active and reactive power of the inverter output, resulting in increasing $V_m$. The bottom panel shows the voltage waveform over the time-window, while the top panel shows a zoomed-in version demonstrating how the voltage amplitude of the signal slowly reaches the new amplitude and frequency setpoints determined by droop control. Our controller is able to accurately track the changing reference, so that the output signal is not able to be distinguished from the reference signal in this plot.}
    \label{fig:droop_vc}
\end{figure}

\section{Discussion and Conclusions}
\label{sec: conclusion}

As the power grid transitions away from fossil fuel-based generation, it is expected that increasing numbers of DC energy sources will supply power to the grid.
The stability results we present are a step to tackle some of the challenges arising for ensured reliability of the grid with increased numbers of CIGs. 

We demonstrate a controller for a half-bridge inverter that explicitly models the switches instead of making time-averaged assumptions. The controller achieves globally, asymptotically stable reference-tracking capabilities, and can handle load changes and a grid-forming control strategy. Furthermore, using a Lyapunov function to guarantee error convergence also suggests further avenues for switching policy design options, for example by allowing a small amount of tracking error in return for reducing the switching frequency and reducing wear on the inverter switches \cite{buisson2005stabilisation}.

Though we believe a hybrid systems approach holds promise for further exploration, we briefly note a few limitations of the inverter control design we developed here. First, our stability result relies on the ability to measure and respond to changes in the sign of $\bold{B}^{\top}\bold{P}e$ fast enough. Furthermore, our stability result is so far only guaranteed for constant resistance loads, whereas most loads include constant inductance or capacitance, or even appear as elements with impedances that vary in response to source voltages. As a step toward resolving this limitation, we have shown the ability of our controller to update in real time to handle changes in resistive load, though this requires knowing what that change in load resistance is. In both of these cases, it may still be possible to provide rigorous stability guarantees after taking realistic sampling frequencies and control delays into account with further development of the theory, however, that is still to be explored. Moreover, ideally we would only have to rely on a sinusoidal voltage reference, however, our proposed control strategy currently requires a current reference as well, which in turn requires knowing the load. Unfortunately, this is not always possible in practice. Finally, our proposed strategy currently does not address parallel voltage sources. We leave considering these exciting research directions to future work.





\bibliographystyle{unsrt}
\bibliography{references}

\end{document}